\documentclass[journal]{IEEEtran}
\IEEEoverridecommandlockouts

\def\hatK{\hat{K}}

\def\boldX{\boldsymbol{X}}

\def\boldU{\boldsymbol{U}}

\def\boldY{\boldsymbol{Y}}

\def\boldV{\boldsymbol{V}}

\def\boldK{\boldsymbol{K}}

\def\ulineS{\underline{S}}
\def\ulineX{\underline{X}}

\def\ulineU{\underline{U}}

\def\ulineW{\underline{W}}

\def\ulineCalS{\underline{\mathcal{S}}}
\def\ulineCalX{\underline{\mathcal{X}}}

\def\ulineCalY{\underline{\mathcal{Y}}}
\def\ulineK{\underline{K}}

\def\ulineY{\underline{Y}}

\def\ulineX{\underline{X}}
\def\ulineY{\underline{Y}}

\def\3To1BC{$3-$to$-1$}

\def\define{:{=}~}

\def\naturals{\mathbb{N}}

\newcommand{\msout}[1]{\text{\sout{\ensuremath{#1}}}}

\newif\ifProofForORDBC
\def\inpsetX{\mathcal{X}_{1}}
\def\inpsetY{\mathcal{X}_{2}}

\def\InpX{{X}_{1}}
\def\InpY{{X}_{2}}
\def\Out{{Y}}
\def\mtimesl{m\times l}
\def\ShrdChnl{\mathbb{W}_{Y_{0}|\underline{U}}}
\newcommand\independent{\protect\mathpalette{\protect\independenT}{\perp}}
\def\independenT#1#2{\mathrel{\rlap{$#1#2$}\mkern2mu{#1#2}}}

\newif\ifJournal
\Journaltrue

\usepackage{amssymb}
\usepackage{amsmath}
\usepackage{mathrsfs}
\usepackage{ulem}
\usepackage{epsf,epsfig}
\usepackage{cite}
\usepackage{color}
\usepackage{dsfont}

\newcommand{\comment}[1]{}
\begin{document}
\sloppy
\newtheorem{remark}{\it Remark}
\newtheorem{thm}{Theorem}
\newtheorem{corollary}{Corollary}
\newtheorem{definition}{Definition}
\newtheorem{lemma}{Lemma}
\newtheorem{example}{Example}
\newtheorem{prop}{Proposition}

\title{Communicating Correlated Sources over an Interference Channel}

\author{\IEEEauthorblockN{Arun Padakandla}
\thanks{This work was supported by the Center for Science of Information (CSoI), an NSF Science and Technology Center, under grant agreement CCF-0939370.}
}
\maketitle
\vspace{-0.5in}
\begin{abstract}
A new coding technique, based on \textit{fixed block-length} codes, is proposed for the problem of communicating a pair of correlated sources over a $2-$user interference channel. Its performance is analyzed to derive a new set of sufficient conditions. The latter is proven to be strictly less binding than the current known best, which is due to Liu and Chen \cite{201112TIT_LiuChe-Shrt}. Our findings are inspired by Dueck's example \cite{198103TIT_Due-Shrt}.
\end{abstract}
\section{Introduction}
\label{Sec:Introduction}
Network information theory has provided us with elegant techniques to exploit correlation amongst distributed information sources. Such a correlation is handled at two levels. Probabilistic (soft) correlation is exploited via binning or transferring them via test channels \cite{198011TIT_CovGamSal-Shrt}. When the sources possess common bits - G\'acs-K\"orner-Witsenhausen (GKW) common part -, conditional coding provides enhanced benefits. In this article, we propose a new coding technique to exploit the presence of \textit{near GKW parts}, amongst distributed sources.

Our primary focus is the scenario depicted in Fig. \ref{Fig:GeneralProblem}. A pair $S_{1},S_{2}$ of correlated sources have to be communicated over a $2-$user interference channel (IC). Receiver (Rx) $j$ wishes to reconstruct $S_{j}$ losslessly. We undertake a Shannon-theoretic study and restrict attention to characterizing sufficient conditions under which $S_{j}$ can be reconstructed at Rx $j$.

The current known best set of sufficient conditions (LC conditions) for this problem is due to Liu and Chen \cite[Thm. 1]{201112TIT_LiuChe-Shrt} and are proven to be optimal for a class of deterministic ICs \cite[Thm. 2]{201112TIT_LiuChe-Shrt}. In this article, we propose a new coding technique based on \textit{fixed block-length} (B-L) codes and derive a new set of sufficient conditions. Through an example (Ex. \ref{Ex:DueckExampleSlightlyModified}), we prove (Lem. \ref{Lem:StrictWeakerConditionsThanLC}) the latter conditions are strictly less binding.

Presence of GKW part enables encoders co-ordinate their inputs, and thereby eliminate interference for the corresponding component of the channel input. Moreover, GKW part enables co-ordination even while enjoying separation. In other words, one can design a channel code corresponding to an optimizing input pmf, unconstrained by the source pmf. 
If $S_{1},S_{2}$ do \textit{not} possess a GKW part, a single-letter (S-L) technique is constrained by the S-L long Markov chain (LMC) $X_{1}-S_{1}-S_{2}-X_{2}$. The S-L LMC can, in general, severely constrain the set of achievable input pmfs (Ex. \ref{Ex:DueckExampleSlightlyModified}, Rem. \ref{Rem:WhyLCIsSub-optimal}). If $S_{1},S_{2}$ possess a \textit{near GKW} part, i.e., $K_{j}=f_{j}(S_{j}): j \in [2]$ such that $\xi=P(K_{1}\neq K_{2})$ is `quite' small, (relatively) large \textit{sub-blocks} of length $l$ could agree with high probability. Indeed, $\xi^{[l]}=P(K_{1}^{l}\neq K_{2}^{l}) = 1-(1-\xi)^{l}\leq l\xi$ can be held small by appropriately choosing $l$. If the encoders employ conditional coding, i.e., identical source to channel mappings, \textit{restricted} to sub-blocks of \textit{fixed length} $l$, then the encoders can enjoy the benefits of separation and co-ordination on a good fraction (at least $\sim(1-\l\xi)$) of these $l-$length sub-blocks. Indeed, we prove in Section \ref{SubSec:CodingForExample} that the latter technique outperforms Liu and Chen's coding technique (LC technique). In Section \ref{Sec:Generalization}, we build on this idea to propose a general coding technique for an arbitrary problem instance.
\begin{figure}
\centering
\includegraphics[width=2.9in]{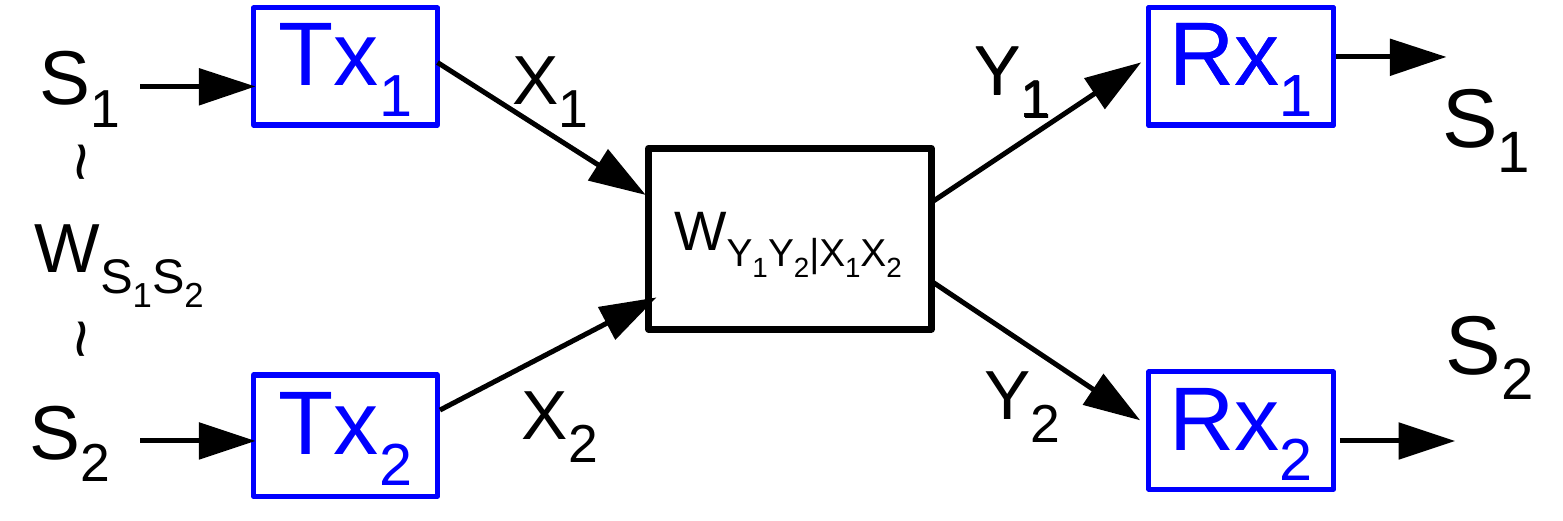}
\caption{Transmission of correlated sources over 2-IC.}
\label{Fig:GeneralProblem}
\end{figure}

Joint source-channel coding over multi-user channels has received considerable attention with regard to characterizing fundamental limits \cite{201604TIT_KheChe-Shrt, 201410TIT_MurDabGun-Shrt, 201607ISIT_WigLap-Shrt} and designing feasible strategies \cite{201504TIT_MinLimKim-Shrt}. Fundamental performance limits for communicating Gaussian sources over Gaussian channels have been studied in \cite{201006TIT_LapTin-Shrt, 201110TIT_TiaDigSha-Shrt}\cite{201508ITWAguGun-Shrt} and the latter considers communication over on IC.

Our findings highlight the sub-optimality of (current known) S-L joint source-channel coding techniques (Rem. \ref{Rem:Multi-LetterCodingScheme}). Notwithstanding this, we derive a \textit{S-L characterization} (Rem. \ref{Rem:Single-LetterCharacterization}) of a new inner bound that \textit{strictly enlarges} the current known best (LC bound). Indeed, the fixed B-L coding technique is an $l-$letter technique. An important second contribution is therefore, a framework - codes and tools (interleaving) - for stitching together S-L techniques in a way that permits performance analysis of the resulting $l-$letter technique and characterization via S-L expressions. Stepping beyond performance characterization, our third contribution is a new coding technique for communicating correlated sources over an IC.

This is part of an evolving work \cite{201607ISIT_Pad, 201601arXiv_Pad, 201701arXiv_Pad} on joint source-channel coding, and is inspired by Dueck's novel example \cite{198103TIT_Due-Shrt} and his \textit{very specific}, yet ingenious, fixed B-L coding. Here, we restrict attention to separation based schemes\footnote{As was done in \cite{201607ISIT_Pad}.} and focus on providing a clear step-by-step description of the ideas. Unifying fixed B-L coding and inducing source correlation onto channel inputs \cite{198011TIT_CovGamSal} involves additional challenges, and is dealt in a concurrent submission \cite[Sec. V]{201601arXiv_Pad}.

\section{Preliminaries : Notation, problem statement}
\label{Sec:Preliminaries}
We let an \underline{underline} denote an appropriate aggregation of related objects. For example, $\underline{S}$ will be used to represent a pair $S_{1},S_{2}$ of RVs. $\underline{\mathcal{S}}$ will be used to denote either the pair $\mathcal{S}_{1},\mathcal{S}_{2}$ or the Cartesian product $\mathcal{S}_{1}\times \mathcal{S}_{2}$, and will be clear from context. 
When $j \in \{1,2\}$, then $\msout{j}$ will denote the complement index, i.e., $\{j,\msout{j}\}=\{1,2\}$. For $m\in \naturals$, $[m]\define \{1,\cdots,m\}$. 
For a pmf $p_{U}$ on $\mathcal{U}$, $b^{*} \in \mathcal{U}$ will denote a symbol with the least positive probability wrt $p_{U}$.\footnote{The underlying pmf $p_{U}$ will be clear from context.} Boldfaces letters such as $\bold{A}$ denote matrices. For a $\mtimesl$ matrix $\bold{A}$, (i) $\bold{A}(t,i)$ denotes the entry in row $t$, column $i$, (ii) $\bold{A}(1:m,i)$ denotes the $i^{th}$ column, $\bold{A}(t,1:l)$ denotes $t^{th}$ row. ``with high probability'', ``single-letter'', ``long Markov chain'', ``block-length'' are abbreviated whp, S-L, LMC, B-L respectively.

For a point-to-point channel (PTP) $(\mathcal{U},\mathcal{Y},\mathbb{W}_{Y|U})$, let $E_{r}(R,p_{U},\mathbb{W}_{Y|U})$ denote the random coding exponent for constant composition codes of type $p_{U}$ and rate $R$ \cite[Thm 10.2]{CK-IT2011-Shrt}. Specifically, $E_{r}(R,p_{U},\mathbb{W}_{Y|U})$ is defined as\vspace{-0.1in}
\begin{eqnarray}
 \label{Eqn:RandomCodingExponent}
 \min_{V_{Y|U}} \left\{D(V_{Y|U}||\mathbb{W}_{Y|U}|p_{U})+|I(p_{U};V_{Y|U})-R|^{+}\right\}.\nonumber
\end{eqnarray}
For RVs $A_{1},A_{2}$, we let $\xi^{[l]}(\underline{A}) \define P(A_{1}^{l}\neq A_{2}^{l})$, and $\xi(\underline{A}) \define \xi^{[1]}(\underline{A})$. Throughout Sec. \ref{Sec:DuecksExample}, $\xi^{[l]}=\xi^{[l]}(\ulineS)$ and $\xi=\xi(\ulineS)$. If $\underline{A}$ is IID, we note\footnote{$(1-x)^{l} \geq 1-xl\mbox{ for }x \in [0,1]$.} $\xi^{[l]}=1-(1-\xi)^{l} \leq l\xi$. We let $\tau_{l,\delta}(K) = 2|\mathcal{K}|\exp\{ -2\delta^{2}p_{K}^{2}(a^{*})l\}$ denote an upper bound on $P(K^{l} \notin T_{\delta}^{l}(K))$ where $T_{\delta}^{l}(K)$ denotes our typical set.

Consider a $2-$user IC with input alphabets $\inpsetX,\inpsetY$, output alphabets $\mathcal{Y}_{1},\mathcal{Y}_{2}$, and transition probabilities $\mathbb{W}_{Y_{1}Y_{2}|\InpX\InpY}$. Let $\ulineS \define (S_{1},S_{2})$, taking values over $\ulineCalS \define \mathcal{S}_{1} \times \mathcal{S}_{2}$ with pmf $\mathbb{W}_{S_{1}S_{2}}$, denote a pair of information sources. For $j\in[2]$, encoder $j$ observes $S_{j}$, and decoder $j$ aims to reconstruct $S_{j}$ with arbitrarily small probability of error (Fig. \ref{Fig:GeneralProblem}). If this is possible, we say $\ulineS$ \textit{is transmissible over IC} $\mathbb{W}_{\underline{Y}|\ulineX}$. In this article, our objective is to characterize sufficient conditions under which $(\ulineCalS,\mathbb{W}_{\ulineS})$ is transmissible over IC $\mathbb{W}_{\underline{Y}|\ulineX}$.
\section{Fixed B-L coding over isolated channels}
\label{Sec:DuecksExample}
We consider a simple generalization (Ex. \ref{Ex:DueckExampleSlightlyModified}) of Dueck's example \cite{198103TIT_Due-Shrt} and propose a coding technique that enables transmissibility of the sources over the corresponding IC. We also prove all current known joint source-channel coding techniques, and in particular LC, is incapable of the same. On the one hand, this proves strict sub-optimality of the latter\footnote{Strict sub-optimality of LC technique can be inferred from \cite{198103TIT_Due-Shrt}. To verify this, modify the MAC therein, to an IC with identical outputs, and use the arguments presented in proof of Lemma \ref{Lem:ExDoesNOTSatisfyLCConditions}. Surprisingly, this has not been documented in \cite{201112TIT_LiuChe-Shrt}.}, and on the other hand, highlights the need for fixed BL coding.
%
\begin{example}
 \label{Ex:DueckExampleSlightlyModified}
 Source alphabets $\mathcal{S}_{1}=\mathcal{S}_{2} = \{0,1,\cdots,a-1 \}^{k}$. Let $\eta\geq 8$ be a positive even integer. The source PMF is
 \begin{eqnarray}
  \label{Eqn:SourceDesc}
  \mathbb{W}_{S_{1}S_{2}}(c^{k},d^{k}) = \begin{cases} \frac{k-1}{k} &\mbox{if } c^{k}=d^{k}=0^{k} \\
\frac{a^{\eta k}-1}{ka^{\eta k}(a^{k}-1)} & \mbox{if } c^{k}=d^{k}, c^{k}\neq 0^{k},\\
\frac{1}{ka^{\eta k}(a^{k}-1)}&\mbox{if } c^{k}=0^{k}, d^{k}\neq 0^{k}\mbox{, and}\end{cases}\nonumber
\end{eqnarray}
$0$ otherwise. Note that in the above eqn. $c^{k},d^{k} \in \mathcal{S}_{1}$ abbreviate the $k$ `digits' $c_{1}c_{2}\cdots c_{k}$ and $d_{1}d_{2}\cdots d_{k}$ respectively. Fig. \ref{Fig:SourceDescription} depicts the source pmf with $\eta = 6$.

The IC is depicted in Fig. \ref{Fig:Step1SetUp} and described below. The input alphabets are $\mathcal{U} \times \mathcal{X}_{1}$ and $\mathcal{U}\times \mathcal{X}_{2}$. The output alphabets are $\mathcal{Y}_{0} \times \mathcal{Y}_{1}$ and $\mathcal{Y}_{0}\times \mathcal{Y}_{2}$. $\mathcal{U} = \mathcal{Y}_{0} = \{0,1,\cdots,a-1\}$. $(U_{j},X_{j}) \in \mathcal{U} \times \mathcal{X}_{j}$ denotes encoder $j$'s input and $(Y_{0},Y_{j}) \in \mathcal{Y}_{0}\times \mathcal{Y}_{j}$ denotes symbols received by decoder $j$. The symbols $Y_{0}$ received at both decoders agree with probability $1$. $\mathbb{W}_{Y_{0}Y_{1}Y_{2}|X_{1}U_{1}X_{2}U_{2}} = \mathbb{W}_{Y_{1}|X_{1}}\mathbb{W}_{Y_{2}|X_{2}}\mathbb{W}_{Y_{0}|U_{1}U_{2}}$, where
\begin{equation}
\label{Eqn:ChnlDesc}
\mathbb{W}_{Y_{0}|U_{1}U_{2}}(y_{0}|u_{1},u_{2}) = \begin{cases} 1 &\mbox{if } y_{0}=u_{1}=u_{2} \\
1 & \mbox{if } u_{1}\neq u_{2}, y_{0}=0,\mbox{ and}
\end{cases}\nonumber
 \end{equation}
$0$ otherwise. The capacities of PTP channels $\mathbb{W}_{Y_{j}|X_{j}}:j=1,2$ are $\mathcal{C}\define h_{b}(\frac{2}{k})+\frac{2}{k}\log a$ and $\mathcal{C} +h_{b}(\frac{2}{ka^{\eta k}})$ respectively.
\end{example}
\begin{figure}
\centering
\includegraphics[width=3.2in]{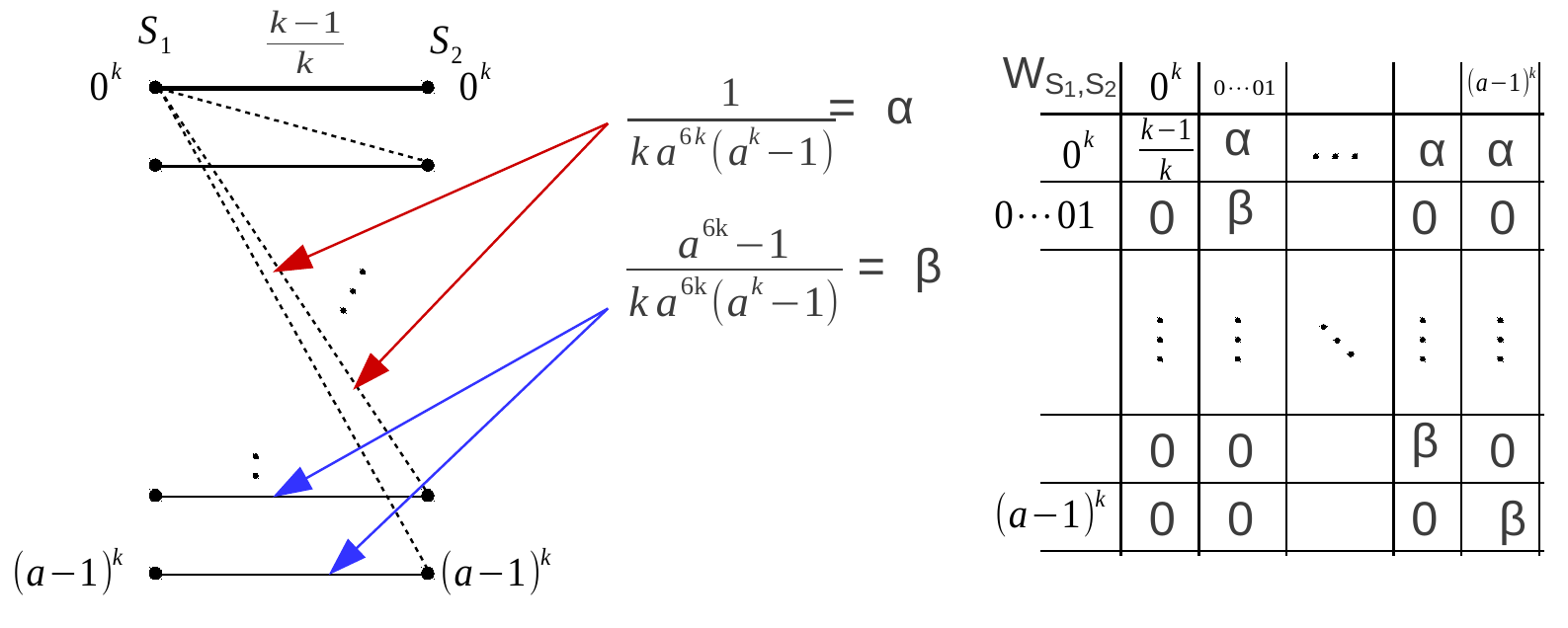}
\caption{On the left, the source pmf is depicted through a bipartite graph. Larger probabilities are depicted through edges with thicker lines. On the right, we depict the probability matrix.}
\label{Fig:SourceDescription}
\end{figure}
\begin{figure}
\centering
\includegraphics[width=2.9in]{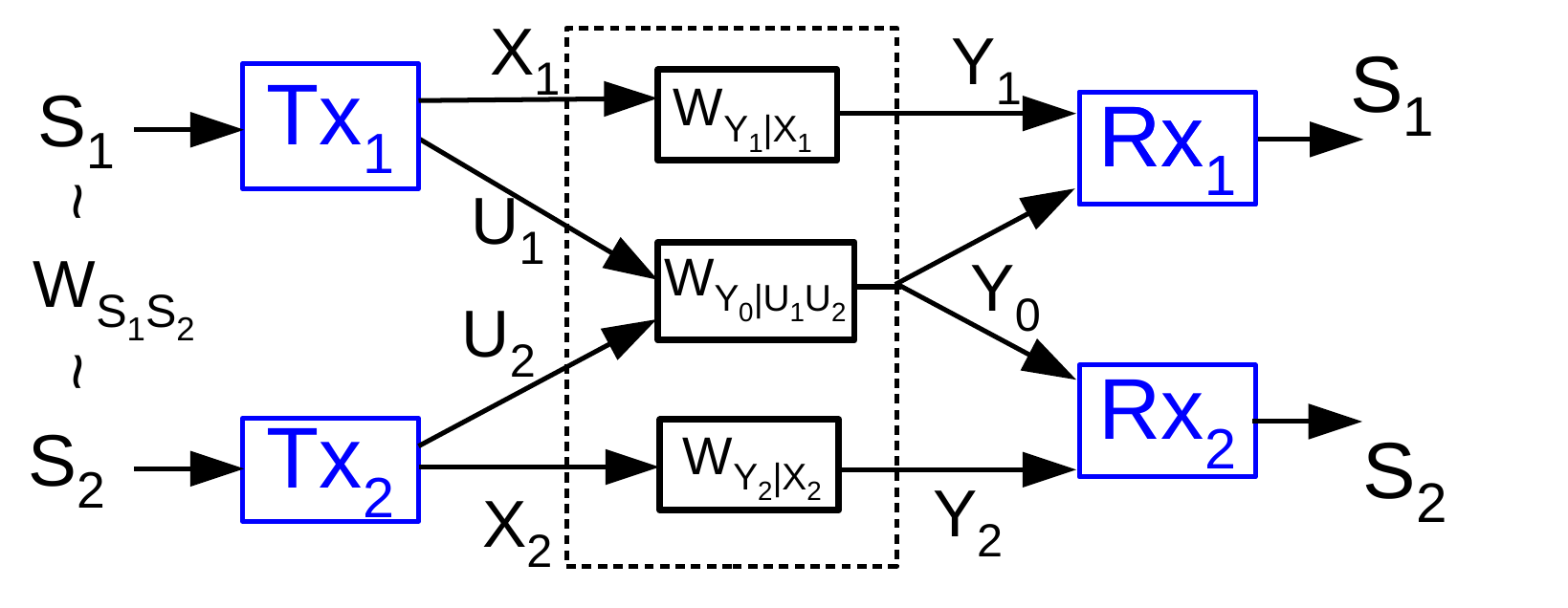}
\caption{Source channel setup of Example \ref{Ex:DueckExampleSlightlyModified}.}
\label{Fig:Step1SetUp}
\end{figure}

We identify key aspects of Ex. \ref{Ex:DueckExampleSlightlyModified}. Let $a,k$ be chosen sufficiently/quiet large. While $\ulineS$ does \textit{not} possess a GKW part, $S_{1}$ and $S_{2}$ agree on most, but not all, realizations. Indeed, $\xi(\ulineS) = \frac{1}{ka^{\eta k}}$ is very small. We also have $H(S_{1}),H(S_{2}),H(\underline{S})$ are\footnote{In fact, $(1-\frac{1}{a^{\eta k}})\log a - \frac{\log 2}{k} \leq \substack{H(S_{1}), H(S_{2}),\\ H(\underline{S})} \leq \log a +h_{b}(\frac{1}{k})+\frac{\log 2}{k}$} $\sim \log a$ and naturally $H(S_{j}|S_{\msout{j}}): j\in [2]$ is very small. Each decoder benefits a lot by decoding \textit{either} source, or any function thereof.
Secondly, $\mathbb{W}_{\ulineS}$ is `\textit{very far}' from the uniform pmf, and hence \textit{any} S-L function $g_{j}(S_{j})$ will remain `considerably' non-uniform.

The IC supports a sum capacity of at most $\log a +2\mathcal{C}+h_{b}(\frac{2}{ka^{\eta k}})$. Since $\sup_{p_{X_{1}X_{2}}}I(\ulineX;\ulineY) = 2\mathcal{C}+h_{b}(\frac{2}{ka^{\eta k}})$, the $\mathbb{W}_{Y_{0}|\ulineU}-$channel must carry bulk of the information (for large $k$). The latter channel carries very little information when $U_{1}\neq U_{2}$, and moreover, it is necessary that $U_{1}=U_{2}$ \textit{and} $U_{1}=U_{2}$ be close to uniform, in order to communicate $\sim \log a$ bits over $\mathbb{W}_{Y_{0}|\ulineU}$.

We first prove Ex. \ref{Ex:DueckExampleSlightlyModified} does not satisfy LC conditions. The proof is based on the following argument. Suppose LC technique enables decoder $j$ reconstruct $S_{j}$ for $j \in [2]$, then both the decoders can reconstruct $S_{1}$ and $S_{2}$ if each of them is provided $Y_{1}$ \textit{and} $Y_{2}$ (and $Y_{0}$). We then prove that this is not permissible, by following an argument similar to \cite[Sec. III.c]{198103TIT_Due-Shrt}.

\begin{lemma}
 \label{Lem:ExDoesNOTSatisfyLCConditions}
Consider Ex. \ref{Ex:DueckExampleSlightlyModified} with any $\eta \in \naturals$. There exists an $a_{*} \in \naturals, k_{*} \in \naturals$, such that for any $a\geq a_{*}$ and any $k\geq k_{*}$, the sources and the IC described in Ex. \ref{Ex:DueckExampleSlightlyModified} do \textit{not} satisfy LC conditions that are stated in \cite[Thm. 1]{201112TIT_LiuChe-Shrt}.
\end{lemma}
The proof is detailed in \cite{201701arXiv_Pad}.
\ifJournal
\begin{proof}
 Since the sources do not have a GKW part, it suffices to prove that Ex. \ref{Ex:DueckExampleSlightlyModified} does not satisfy conditions stated in \cite[Corollary 1]{201112TIT_LiuChe-Shrt}. Let $\ulineS Q\ulineW\ulineX\ulineU\ulineY$ be \textit{any} collection of RVs whose pmf factorizes as $\mathbb{W}_{\underline{S}}p_{Q}p_{W_{1}|Q}p_{W_{2}|Q}p_{X_{1}U_{1}|S_{1}W_{1}Q}p_{X_{2}U_{2}|S_{2}W_{2}Q}\mathbb{W}_{\underline{Y}|\underline{X}\underline{U}}$. We prove
 \begin{eqnarray}
H(\underline{S}) \!> \!I(S_{1}X_{1}U_{1};Y_{1}Y_{0}|Q\underline{W})\!+\!I(W_{1}S_{2}X_{2}U_{2};Y_{2}Y_{0}|Q)\nonumber\\
\label{Eqn:Eqn44OfCorollary1LiuChen}
\!\!\!\!\!-I(S_{1};S_{2})
 \end{eqnarray}
and thereby contradicting \cite[Eqn. (44), Corollary 2]{201112TIT_LiuChe-Shrt}. Towards that end, we first note \begin{eqnarray}
H(\underline{S}) &\geq& H(S_{2}) =h_{b}(\frac{1}{k})+\frac{1}{k}\log (a^{k}-1) \geq \frac{1}{k}\log\frac{ka^{k}}{2}\nonumber\\\
&\geq& \log a +\frac{1}{k}\log (\frac{k}{2}) \geq \log a,\nonumber
\end{eqnarray}
whenever $a^{k} \geq 2$. Secondly, the RHS of (\ref{Eqn:Eqn44OfCorollary1LiuChen}) can be bounded above by
\begin{eqnarray}
 \label{Eqn:SumRateOfICIsSmallerThanMACSumRate}
 \!\!\!\!\!\!\!\!\!\!\!\!\lefteqn{\!\!\!\!\!\!\!\!\!\!\!\!\!I(S_{1}X_{1}U_{1};Y_{1}Y_{0}|Q\underline{W})\!+\!I(W_{1}S_{2}X_{2}U_{2};Y_{2}Y_{0}|Q)-I(S_{1};S_{2})}\nonumber\\
 \!\!\!\!\!\!\!\!\!\!\!\!\!\!\!\! \lefteqn{\!\!\!\!\!\!\!\!\!\!\!\!\leq I({S_{1}X_{1}U_{1};Y_{1}Y_{0}|Q\underline{W}})+I({\underline{W}S_{2}X_{2}U_{2};\underline{Y}|Q})- I({S_{1};S_{2}})} \nonumber\\
 \!\!\!\!&\leq& I({\underline{WSXU};\underline{Y}|Q})+I({S_{1}X_{1}U_{1};\underline{Y}|Q\underline{W}})\nonumber\\
 \!\!\!\!&&-I({S_{1}X_{1}U_{1};\underline{Y}|Q\underline{W}S_{2}X_{2}U_{2}})- I({S_{1};S_{2}}) \nonumber\\
 \!\!\!\!&=& I({\underline{XU};\underline{Y}|Q})+I({S_{1}X_{1}U_{1};\underline{Y}|Q\underline{W}})\nonumber\\
 \label{Eqn:ChnlAndSrcMarkovChain}
 \!\!\!\!&&-I({S_{1}X_{1}U_{1};\underline{Y}X_{2}U_{2}|Q\underline{W}S_{2}})- I({S_{1};S_{2}}) \\
 \label{Eqn:SrcsIndependentOfW}
 \!\!\!\!&=& I({\underline{XU};\underline{Y}|Q})+I({S_{1}X_{1}U_{1};\underline{Y}|Q\underline{W}})\nonumber\\&&-I({S_{1}X_{1}U_{1};\underline{Y}X_{2}U_{2}S_{2}|Q\underline{W}})\nonumber\\
 &\leq& I({\underline{XU};\underline{Y}|Q}) \nonumber
 \end{eqnarray}
where (\ref{Eqn:ChnlAndSrcMarkovChain}) follows from Markov chains $Q\underline{SW}-\underline{XU}-\underline{Y}$ and $S_{1}X_{1}U_{1}-S_{2}Q\underline{W} - X_{2}U_{2}$ and (\ref{Eqn:SrcsIndependentOfW}) follows from the independence of $\underline{S}$ and $\underline{W}Q$.

We now argue the RHS above is strictly lesser than $\log a$. Our argument pretty much follows Dueck \cite[Sec. III.c]{198103TIT_Due-Shrt} verbatim and is provided here for the sake of completeness. Let $R=\mathds{1}_{\{(S_{1},S_{2})=(0^{k},0^{k}) \}}$.
\begin{eqnarray}
 \label{Eqn:DuecksArgumentPart1}
 I({\underline{XU};\underline{Y}|Q}) \leq I({\underline{XU}R;\underline{Y}|Q}) \leq \log 2+I({\underline{XU};\underline{Y}|Q,R})\nonumber\\
 \leq \log 2 + \frac{\log |\mathcal{Y}_{0}\times\mathcal{Y}_{1}\times\mathcal{Y}_{2}|}{k} +(1-\frac{1}{k})I({\underline{XU};\underline{Y}|Q,R=1})\nonumber
\end{eqnarray}
We focus on the third term in the above sum. Conditioned on $R=1$, the sources are equal to $(0^{k},0^{k})$. It can be verified that $X_{1}U_{1}-S_{1}Q-S_{2}Q-X_{2}U_{2}$. Given $Q=q,R=1$, $(X_{1},U_{1})$ is independent of $(X_{2},U_{2})$ and hence 
\begin{eqnarray}
I({\underline{XU};\underline{Y}|Q,R=1}) \leq \max_{p_{X_{1}U_{1}}p_{X_{2}U_{2}}} I({\underline{XU};\underline{Y}|Q,R=1})\nonumber\\
\label{Eqn:Step3}
\leq 2\mathcal{C}+h_{b}(\frac{2}{ka^{\eta k}})+\max_{p_{U_{1}}p_{U_{2}}\mathbb{W}_{Y_{0}|\ulineU}} H(Y_{0})
\end{eqnarray}
We now evaluate an upper bound on the maximum value of $H(Y_{0})$ subject to $U_{1},U_{2}$ being independent. We evaluate the following three possible cases.

Case 1a : For some $u\in\mathcal{U}$, $P(U_{1}=u)\geq \frac{1}{2}$ and $P(U_{2}=u)\geq \frac{1}{2}$. Then $P(Y_{0}=u)\geq \frac{1}{4}$ (independence of $U_{1},U_{2}$) and hence $H(Y_{0})\leq \log 2 + \frac{3}{4}\log a$.

Case 1b : For some $u\in\mathcal{U}$, $P(U_{1}=u)\geq \frac{1}{2}$ and $P(U_{2}=u)\leq \frac{1}{2}$. Then $P(U_{2}\neq u)\geq \frac{1}{2}$ and hence $P(Y_{0}=0) \geq \frac{1}{4}$ and hence $H(Y_{0})\leq \log 2 + \frac{3}{4}\log a$.

Case 2a : For every $u \in \mathcal{U}$, $P(U_{1}=u)\leq \frac{1}{2}$. Then for any $u \in \mathcal{U}$, $P(U_{2} \neq U_{1}) = \sum_{u}\sum_{z\neq u}P(U_{2}=u)P(U_{1}=z) \geq \frac{1}{2}\sum_{u}P(U_{2}=u)=\frac{1}{2}$, implying $P(Y_{0}=0)\geq \frac{1}{2}$ and hence $H(Y_{0})\leq \log 2 + \frac{3}{4}\log a$.

In all cases, we have $H(Y_{0})\leq \log 2 + \frac{3}{4}\log a$. Substituting through (\ref{Eqn:Step3}) and above, we conclude
\begin{eqnarray}
 \label{Eqn:}
 I({\underline{XU};\underline{Y}|Q}) \leq 2\log2+2\mathcal{C}+h_{b}(\frac{2}{ka^{\eta k}})+\frac{3}{4}\log a + \frac{\log |\ulineCalY|}{k}\nonumber\\
 < \log a\nonumber
\end{eqnarray}
for sufficiently large $k,a$.
\end{proof}
\fi
\begin{remark}
\label{Rem:WhyLCIsSub-optimal}
Why is the LC technique incapable of communicating $\ulineS$? Any valid pmf $p_{U_{1}U_{2}}$ induced by a S-L coding scheme is constrained to the LMC $U_{1}-S_{1}-S_{2}-U_{2}$. For $j\in [2]$, $p_{U_{j}|S_{j}}$ can equivalently be viewed as $U_{j}=g_{j}(S_{j},W_{j})$, for some function $g_{j}$ and RV $W_{j}$, that satisfy $W_{1} \independent W_{2}$. Owing to the latter, $W_{1}$ and/or $W_{2}$ being \textit{non-trivial} RVs, reduces $P(U_{1}=U_{2})$. If we let, $W_{1},W_{2}$ be deterministic, the only way to make $U_{j}$ uniform is to pool less likely symbols. However, the source is `highly' non-uniform, and even by pooling \textit{all} the less likely symbols, we can gather a probability, of at most, $\frac{1}{k}$. Consequently, any $p_{U_{1}U_{2}}$ induced via a S-L coding scheme is sufficiently far from any pmf that satisfies $U_{1}=U_{2}$ whp \textit{and} $U_{1}=U_{2}$ close to uniform.
\end{remark}
\begin{remark}
 \label{Rem:Multi-LetterCodingScheme} An $l-$letter (multi-letter with $l>1$) coding scheme is constrained by an $l-$letter LMC $U_{1}^{l}-S_{1}^{l}-S_{2}^{l}-U_{2}^{l}$.\footnote{Constraint will \textit{not} be via a S-L LMC.} Suppose we choose $l$ \textit{reasonably} large such that 1) $\xi^{[l]}(\ulineS)$ is not high, and 2) $S_{j}^{l}$ is \textit{reasonably} uniform on its typical set $T_{\delta}^{l}(S_{j})$, and define $U_{j}:j \in [2]$ through \textit{identical} functions $U_{j}^{l}=g(S_{j}^{l}):j \in [2]$, then one can easily visualize the existence of $g$ such that $p_{U_{1}^{l}U_{2}^{l}}$ satisfies the twin objectives of $U_{1}^{l}=U_{2}^{l}$ whp and $U_{1}^{l}=U_{2}^{l}$ is close to uniform. Our coding scheme, will in fact, identify such $g$ maps. This portrays the sub-optimality of S-L schemes for joint source-channel coding.
\end{remark}


\subsection{Fixed Block-length Coding over a Noiseless Channel}
\label{SubSec:CodingForExample}
In order to input codewords on the $\mathbb{W}_{Y_{0}|\underline{U}}-$channel, that agree, we employ the same source code, same channel code and same mapping, each of \textit{fixed B-L}\footnote{irrespective of the desired prob. of error} $l$, at both encoders. $l$ is chosen large enough such that the source can be reasonably efficiently compressed, and yet small enough, to ensure $\xi^{[l]}(\ulineS)$ is reasonably small. We refer to these $l-$length blocks as \textit{sub-blocks}. Since $l$ is fixed, there is a non-vanishing probability that these source sub-blocks will be decoded erroneously. An outer code, operating on an arbitrarily large number $m$ of these sub-blocks, will carry information to correct for these `errors'. The outer code will operate over satellite channel $\mathbb{W}_{Y_{j}|X_{j}}$. We begin with a description of the fixed B-L codes.

We employ a simple fixed B-L (inner) code. Let $T_{\delta}^{l}(S_{1})$ be the source code, and let $C_{U}=\mathcal{U}^{l}$ be the channel code. Let $lA =\lfloor\log a^{l}\rfloor$ bits, of the $\lceil\log |T_{\delta}^{l}(S_{1})| \rceil$ bits output by the source code, be mapped to $C_{U}$. Both encoders use the same source code\footnote{Encoder $2$ also employs source code $T_{\delta}^{l}(S_{1})$, (and not $T_{\delta}^{l}(S_{2})$).}, channel code and mapping. 

Suppose we communicate an \textit{arbitrarily large} number $m$ of these sub-blocks on $\ShrdChnl$ as above. Moreover, suppose encoder $j$ communicates the rest of the $lB = \lceil\log |T_{\delta}^{l}(S_{1})| \rceil-lA$ bits output by its source code to decoder $j$ on its satellite channel $\mathbb{W}_{Y_{j}|X_{j}}$.\footnote{Through our description, we assume communication over $\mathbb{W}_{Y_{j}|X_{j}}$ is noiseless. In the end, we prove that the rate we demand of $\mathbb{W}_{Y_{j}|X_{j}}$ is lesser than its capacity, justifying this assumption.} How much more information needs to be communicated to decoder $j$, to enable it reconstruct $S_{j}^{lm}$? We do a simple analysis that suggests a natural coding technique.

View the $m$ sub-blocks of $S_{j}$ as the rows of the matrix $\bold{S}_{j}(1:m,1:l) \in \mathcal{S}_{j}^{m \times l}$. Let $\bold{\hatK}_{j}(1:m,1:l) \in \mathcal{S}_{1}^{ m \times l}$ denote decoder $j$'s reconstruction of $\bold{S}_{1}(1:m,1:l)$\footnote{1) Encoder $j$ could input any arbitrary $C_{U}-$codeword when its sub-block $S_{j}^{l} \notin T_{\delta}^{l}(S_{1})$, and decoder $j$ could declare an arbitrary reconstruction when it observes $Y_{0}^{l}=0^{l}$. 2) $\bold{\hatK}_{2}(1:m,1:l)$ is also viewed as reconstruction of $\bold{S}_{1}(1:m,1:l)$.}. The $m$ sub-blocks
\begin{equation}
 \label{Eqn:Sub-BlocksAreIID}
 \left\{ \left(\bold{S}_{j}(t,1:l),\bold{\hatK}_{j}(t,1:l):j=1,2\right):t \in [m] \right\}
\end{equation}
are iid\footnote{The $l-$length mappings from the source sub-blocks to the $l-$length channel codewords input on the $\mathbb{W}_{Y_{0}|\underline{U}}$ are identical across the sub-blocks and the source sub-blocks are mutually independent.} with an $l-$length distribution $\mathbb{W}_{S_{1}^{l}S_{2}^{l}}p_{\hatK_{1}^{l}\hatK_{2}^{l}|S_{1}^{l}S_{2}^{l}}$\footnote{$p_{\hatK_{1}^{l}\hatK_{2}^{l}|S_{1}^{l}S_{2}^{l}}$ does not necessarily factor, owing to the $l-$length maps.}. Since, in principle, we can operate by treating these $l-$length sub-blocks as a \textit{super-symbol}, and employ standard binning technique over these $m$ super-symbols, decoder $j$ needs only $H(S_{j}^{l}|\hat{K}_{j}^{l})$ bits per source sub-block. We have no characterization of $p_{\hatK_{1}^{l}\hatK_{2}^{l}|S_{1}^{l}S_{2}^{l}}$, and hence we derive an upper bound.
\begin{eqnarray}
 H(S_{j}^{l}|\hat{K}_{j}^{l}) \leq H(S_{j}^{l},\mathds{1}_{\{\hat{K}_{j}^{l} \neq S_{1}^{l}\}}|\hat{K}_{j}^{l}) \leq h_{b}(P(\hat{K}_{j}^{l} \neq S_{1}^{l}))+\nonumber\\
 \label{Eqn:BoundOnAddInformation}
\!\!\!+P(\hat{K}_{j}^{l} \neq S_{1}^{l})\log |\mathcal{S}_{j}^{l}|+P(\hat{K}_{j}^{l} = S_{1}^{l})H(S_{j}^{l}|S_{1}^{l}).\!\!\!\\
\leq l\mathcal{L}^{S}_{l}(P(\hat{K}_{j}^{l} \neq S_{1}^{l}),|\mathcal{S}_{j}|)+lH(S_{j}|S_{1}),\mbox{ where}\nonumber\\
\label{Eqn:AdditionalInformationToGoOnSatelliteChannels}
 \mathcal{L}_{l}^{S}(\phi,|\mathcal{K}|) \define \frac{1}{l}h_{b}(\phi) + \phi\log|\mathcal{K}|,~~~~~~~
\end{eqnarray}
represents the additional source coding rate needed to compensate for the errors in the fixed B-L decoding. It suffices to prove $\mathcal{L}^{S}_{l}(P(\hat{K}_{j}^{l} \neq S_{1}^{l}),|\mathcal{S}_{j}|)+B+H(S_{j}|S_{1})$ $<$ capacity of $\mathbb{W}_{Y_{j}|X_{j}}$. Since $\mathcal{L}_{l}^{S}(\phi,|\mathcal{K}|)$ is non-decreasing in $\phi$ if $\phi < \frac{1}{2}$, we bound $P(\hat{K}_{j}^{l} \neq S_{1}^{l})$ by a quantity that is less than $\frac{1}{2}$. Towards that end, note that $\{ S_{1}^{l} \neq \hat{K}_{j}^{l}\} \subseteq \{ S_{1}^{l}\neq S_{2}^{l}\} \cup \{ S_{1}^{l} \notin T_{\delta}^{l}(S_{1})\}$. Indeed, $S_{1}^{l}=S_{2}^{l} \in T_{\delta}^{l}(S_{1})$ implies both encoders input same $C_{U}-$codeword and agree on the $lB$ bits communicated to their respective decoders. Therefore $P( S_{1}^{l} \neq \hat{K}_{j}^{l})\leq \phi$, where $\phi=\tau_{l,\delta}(S_{1})+\xi(\ulineS)$, 
\begin{eqnarray}\label{Eqn:DefnOfTauLDelta}\!\!\!\!\!\!\tau_{l,\delta}(S_{1}) \!\!\!\!\!&\leq& 2a^{k}\exp\{ -\frac{\delta^{2}l}{2k^{2}a^{2k}}\}\mbox{ and }\xi^{[l] }(\ulineS) \leq \frac{l}{ka^{\eta k}}.
\end{eqnarray}
Choose $l= k^{4}a^{\frac{\eta k}{2}}, \delta = \frac{1}{k}$, substitute in (\ref{Eqn:DefnOfTauLDelta}). Since $\eta\geq 6$, verify  $\phi \leq 2k^{3}a^{-\frac{\eta k}{2}}<\frac{1}{2}$ for sufficiently large $a,k$. Verify\vspace{-0.1in} \begin{equation}\label{Eqn:FinalBoundOnExtraSourceCodingRate}\mathcal{L}^{S}_{l}(2k^{3}a^{-\frac{\eta k}{2}},|\mathcal{S}_{j}|) \leq \frac{1}{4k}\log a\end{equation}for sufficiently large $a,k$. Recall $lB = \lceil\log |T_{\delta}^{l}(S_{1})| \rceil-lA$. Substituting $\delta=\frac{1}{k}$, verify\footnote{Use $H(S_{1})\leq \log a +h_{b}(\frac{1}{k})$ and $|T_{\delta}(S_{1})|\leq \exp\{l(1+\delta)H(S_{1})\}$.}
 \begin{equation}
 \label{Eqn:UpperBoundOnB}
 B \leq (2/l)+(1/k)\log a + (1+(1/k))h_{b}(1/k).
 \end{equation}
Since $h_{b}(\frac{2}{k})-(1+\frac{1}{k})h_{b}(\frac{1}{k}) \geq \frac{1}{2k}\log\frac{k}{256}$ for large enough $k$, RHS of (\ref{Eqn:FinalBoundOnExtraSourceCodingRate}), (\ref{Eqn:UpperBoundOnB}) sum to at most $\frac{2}{l}+h_{b}(\frac{2}{k})+\frac{5}{4k}\log a$ for large enough $a,k$. Furthermore, $H(S_{2}|S_{1})\leq h_{b}(\frac{1}{(k-1)a^{\eta k}})+\frac{2}{a^{\eta k}}\log a \leq h_{b}(\frac{2}{ka^{\eta k}})+\frac{1}{4k}\log a$ for sufficiently large $a,k$. It can now be easily verified that the satellite channels support these rates for large enough $a,k$.

A few details with regard to the above coding technique is worth mentioning. 
$p_{\hatK_{1}^{l}\hatK_{2}^{l}|S_{1}^{l}S_{2}^{l}}$ can in principle be computed, once the fixed block-length codes, encoding and decoding maps are chosen. $S_{j}^{lm}$ will be binned at rate $H(S_{j}^{l}|\hatK_{j}^{l})$ and the  decoder can employ a joint-typicality based decoder using the computed $p_{S_{j}^{l}|\hatK_{j}^{l}}$. We conclude the following.

\begin{thm}
 \label{Thm:StrictSub-OptimalityOfLC}
 The LC conditions stated in \cite[Thm. 1]{201112TIT_LiuChe-Shrt} are not necessary. Refer to Ex. \ref{Ex:DueckExampleSlightlyModified}. There exists $a^{*} \in \naturals$ and $k^{*}\in \naturals$ such that for any $a\geq a^{*}$ and any $k \geq k^{*}$, $S_{1},S_{2}$ and the IC $\mathbb{W}_{\underline{Y}|\underline{X}\underline{U}}$ do not satisfy LC conditions, and yet, $\ulineS$ is transmissible over IC $\mathbb{W}_{\underline{Y}|\underline{X}\underline{U}}$.
\end{thm}
\begin{remark}
 \label{Rem:PerformanceWithl} The above scheme crucially relies on the choice of $l$ - neither too big, nor too small. This is elegantly captured as follows. As $l$ increases, $\xi^{[l]}(\ulineS)\rightarrow 1$, $\tau_{l,\delta}$ (and $g_{\rho,l}$)\footnote{$g_{\rho,l}$ is defined in the sequel.} $\rightarrow 0$. As $l$ decreases, $\xi^{[l]}(\ulineS) \rightarrow \xi(\ulineS)$, and $\tau_{l,\delta}$ (and $g_{\rho,l}$) $\rightarrow 1$. If $\phi \rightarrow 0.5$, $\mathcal{L}^{S}_{l}(\phi,|\mathcal{S}_{j}|)\rightarrow 0.5\log|\mathcal{S}_{j}|=\frac{k}{2}\log a$.
\end{remark}

\section{Fixed BL codes over an arbitrary IC}
\label{Sec:Generalization}
Our analysis (Sec. \ref{SubSec:CodingForExample}) focused on proving\vspace{-0.05in}
\begin{equation}
 \label{Eqn:ConditionCentralToSeparation}
 \mathcal{L}^{S}_{l}(\phi,|\mathcal{S}_{j}|)+B+H(S_{j}|S_{1}) \leq I(X_{j};Y_{j})\vspace{-0.05in}
\end{equation}
where $\phi < \frac{1}{2}$ was an upper bound on $P(\hat{K}_{j}^{l} \neq S_{1}^{l})$. All our sufficient conditions will take this form. The lack of isolation between channels carrying fixed B-L and infinite B-L codes will throw primarily two challenges.\footnote{And an additional loss in the channel rate, denoted $\mathcal{L}^{C}(\cdot,\cdot)$.} We present our generalization in three pedagogical steps.

In general, $P(S_{1}^{l}\neq \hat{K}_{j}^{l})\leq\tau_{l,\delta}+\xi^{[l]}+g_{\rho,l}$, where the first two terms are as in (\ref{Eqn:DefnOfTauLDelta}), and $g_{\rho,l}$ is the probability that any of the decoders incorrectly decodes the $C_{U}-$codeword, conditioned on both encoders choosing the \textit{same} $C_{U}$ codeword.\footnote{For Ex.\ref{Ex:DueckExampleSlightlyModified}, $g_{\rho,l}=0$, and we ignored it. For general IC, $g_{\rho,l}$ is non-zero.} Our fixed B-L code $C_{U}$ will be a constant composition code, and in the statements of all theorems, $g_{\rho,l}$ is defined as\vspace{-0.1in}
\begin{equation}
\label{Eqn:DefnOfgrhol}
g_{\rho,l} \define \sum_{j=1}^{2}\exp\{ -l(E_{r}(A+\rho,p_{U},p_{Y_{j}|U})-\rho)\}.\nonumber\vspace{-0.1in}
\end{equation}
In all our theorems, $\mathcal{L}_{l}^{S}(\phi,|\mathcal{S}_{j}|)$ is defined as in (\ref{Eqn:AdditionalInformationToGoOnSatelliteChannels}), $g_{\rho,l}$ as above, $\phi=\tau_{l,\delta}(K_{1})+\xi^{[l]}(\ulineK)+g_{\rho,l}$ will serve as an upper bound on $P(S_{1}^{l}\neq \hat{K}_{j}^{l})$ that is less than $\frac{1}{2}$.
%
%
\subsection{Designing independent streams ignoring self-interference}
\label{SubSec:IgnoringSelfInterference}
The main challenges in generalizing pertain to 1) multiplexing a fixed B-L code with an infinite B-L code through a single channel input, and 2) the effect of erroneous conditional coding on the outer code. We adapt tools developed by Shirani and Pradhan \cite{201406ISIT_ShiPra-2-Shrt,201307ISIT_ChaSahPra-Shrt} in the context of distributed source coding. The following very simple generalization is chosen to illustrate our ideas. In particular, we live with self-interference between the two streams.
\begin{thm}
 \label{Thm:FirstStep}
$(\ulineCalS,\mathbb{W}_{\ulineS})$ is transmissible over IC $(\ulineCalX,\underline{\mathcal{Y}},\mathbb{W}_{\ulineY|\ulineX})$ if there exists (i) a finite set $\mathcal{K}$, maps $f_{j}:\mathcal{S}_{j}\rightarrow \mathcal{K}$, with $K_{j}=f_{j}(S_{j})$ for $j\in[2]$, (ii) $l \in \naturals, \delta > 0$, (iii) finite set $\mathcal{U},\mathcal{V}_{1},\mathcal{V}_{2}$ and pmf $p_{U}p_{V_{1}}p_{V_{2}}p_{X_{1}|UV_{1}}p_{X_{2}|UV_{2}}$ defined on $\mathcal{U}\times \underline{\mathcal{V}}\times \ulineCalX$, where $p_{U}$ is a type of sequences\footnote{Please refer to \cite[Defn 2.1]{CK-IT2011-Shrt}} in $\mathcal{U}^{l}$, (iv) $A,B \geq 0$, $\rho \in (0,A)$ such that $\phi \in [0,0.5)$, \vspace{-0.05in}
 \begin{eqnarray}
 \label{Eqn:TypicalSetSize}
A+B \geq (1+\delta)H(K_{1}),\mbox{ and for }j \in [2],\nonumber\\
B+H(S_{j}|K_{1}) + \mathcal{L}^{S}_{l}(\phi,|\mathcal{S}_{j}|) < I(V_{j};\Out_{j})-\mathcal{L}_{j}^{C}(\phi,|\mathcal{V}|),\nonumber
\end{eqnarray}
where, $\phi \define g_{\rho,l}+\xi^{[l]}(\ulineK)+\tau_{l,\delta}(K_{1})$, $\mathcal{L}_{j}^{C}(\phi,|\mathcal{U}|) = h_{b}(\phi)+\phi\log |\mathcal{U}|+|\mathcal{Y}||\mathcal{U}|\phi\log\frac{1}{\phi}$.
\end{thm}
\begin{remark}
 \label{Rem:Single-LetterCharacterization}
 The characterization provided in Thm. \ref{Thm:FirstStep} (and those in Thm. \ref{Thm:SecondStep}, \ref{Thm:ThirdStep})) is via S-L PMFs and S-L expressions.
\end{remark}

\begin{remark}
 \label{Rem:ExplainThm}
In the above, the fixed B-L code operates over $K_{j}^{l}$ instead of $S_{j}^{l}$. $\mathcal{L}_{j}^{C}(\phi,|\mathcal{U}|)$ quantifies the loss in rate of the outer code due to erroneous conditional coding. Note that, in Ex. \ref{Ex:DueckExampleSlightlyModified}, satellite channel remained unaffected when the encoders placed different $C_{U}$ codewords. The latter events imply, the $V_{j}-Y_{j}$ channel is not $p_{Y_{j}|V_{j}}$. $\mathcal{L}_{j}^{C}(\phi,|\mathcal{U}|)$ is a bound on the difference in the mutual information between $p_{Y_{j}|V_{j}}$ and the actual channel. Note that $\mathcal{L}_{j}^{C}(\phi,|\mathcal{U}|) \rightarrow 0$ as $\phi \rightarrow 0$.
\end{remark}
\begin{proof}
We elaborate on the new elements. The rest follows from standard arguments \cite{201701arXiv_Pad}. The source-coding module, and the mappings to the channel-coding module are identical to Section \ref{SubSec:CodingForExample}. We describe the structure of $C_{U}$ and how it is multiplexed with the outer code built on $\mathcal{V}_{j}$. If we build a single code $C_{V_{j}}$ of B-L $lm$ and multiplex it with $m$ blocks of $C_{U}$, then $C_{V_{j}}$ does \textit{not} experience an IID memoryless channel.

Let $\boldU_{j}(t,1:l)$ denote encoder $j$'s chosen codeword from $C_{U}$ corresponding to the $t^{\mbox{{\small th}}}$ sub-block of  $\boldK_{j}(t,1:l)$. We seek to identify sub-vectors of $\boldU_{j}$ that are IID, and whose pmf we know. We can then multiplex the outer code along these sub-vectors. \textit{Interleaving} \cite{201307ISIT_ChaSahPra-Shrt} enables us do this. 

Suppose, for $t \in [m]$, $\bold{A}(t,1:l) \sim p_{A^{l}}$ and the $m$ vectors $\bold{A}(1,1:l),\cdots,\bold{A}(m,1:l)$ are iid $p_{A^{l}}$. Let $\pi_{t}:[l] \rightarrow [l]:t\in [m]$ be $m$ surjective maps, that are independent and uniformly chosen among the collection of surjective maps (permuters). Then, for each $i \in [l]$, the $m-$length vector\vspace{-0.1in}
\[
\bold{A}(1,\pi_{1}(i)),\bold{A}(2,\pi_{2}(i)), \cdots, \bold{A}(m,\pi_{m}(i))\sim \prod_{t=1}^{m}\{ \frac{1}{l}\sum_{i=1}^{l}p_{A_{i}}\}.
\]
\cite[Appendix A]{201601arXiv_Pad} contains a proof. The following notation will ease exposition. For $\bold{A} \in \mathcal{A}^{m \times l}$, and a collection $\pi_{t}:[l] \rightarrow [l]: t \in [m]$ of surjective maps, we let $\bold{A}^{\pi} \in \mathcal{A}^{ m\times l}$ be such that $\bold{A}^{\pi}(t,i) \define \bold{A}(t,\pi_{t}(i))$ for each $(t,i)\in [m]\times [l]$. The above fact can be therefore be stated as $\bold{A}^{\pi}(1:m,i) \sim \prod_{t=1}^{m}\{ \frac{1}{l}\sum_{i=1}^{l}p_{A_{i}}\}$. 

One can now easily prove that, if $C_{U}$ is a \textit{constant composition code} of type $p_{U}$, and $m$ codewords are independently chosen\footnote{Not necessarily uniformly. In fact the index output from the source code, owing to the fixed block-length $l$ is not necessarily uniform} from $C_{U}$ and placed as rows of $\boldU_{j}$, then for any $i \in [l]$, the \textit{interleaved} vector $\bold{U}^{\pi}(1:m,i) \sim \prod_{t=1}^{m}p_{U}$. 
We now build $l-$codebooks (independently drawn), one for each of these \textit{interleaved} vectors.

Following is our channel code structure. $C_{U}$ is constant composition code of type $p_{U}$ and B-L $l$. Encoder $j$ picks $l$ independent codes $C_{V_{ji}}: i \in [l]$, each iid $\sim \prod_{t=1}^{m}p_{V_{j}}(\cdot)$, each B-L $m$. $C_{V_{ji}}$ is multiplexed along with sub-vector $\boldU_{j}^{\pi}(1:m,i)$. Outer code message is split into $l$ equal parts $(M_{j1},\cdots, M_{jl})$. $\boldV_{j} \in \mathcal{V}^{m \times l}$ is defined as $\boldV_{j}^{\pi}(1:m,i) = C_{V_{ji}}(M_{ji})$. For each $(t,i) \in [m] \times [l]$, $\boldX_{j}(t,i)$ is chosen IID wrt $p_{X_{j}|UV_{j}}(\cdot|\boldU_{j}(t,i),\boldV_{j}(t,i))$. Symbols in $\boldX_{j}\in \mathcal{X}_{j}^{m \times l}$ are input on the channel. It can be verified that (1) the codewords of $C_{U}$ and $C_{V_{ji}}$ pass through an IID memoryless channels whose transition probabilities are `characterized' in the sequel.
 
Since each codebook $C_{V_{ji}}$ and each codeword is IID, $\boldU_{j}(t,1:l)\rightarrow \boldY_{j}(t,1:l)$ is IID $p_{Y_{j}|U}-$channel. Interleaving ensures $\boldV_{j}^{\pi}(1:m,i)\rightarrow \boldY_{j}^{\pi}(1:m,i)$ is IID. But, unless $\boldU_{1}=\boldU_{2}$, we are not guaranteed the latter channel is $\prod p_{Y_{j}|V_{j}}$.\footnote{We are unaware of the transition probabilities of this IID PTP.} In fact, we only know certain marginals of 
\begin{equation}\scalebox{0.89}{$
{\boldU^{\pi}_{j}(1:m,i),\boldV^{\pi}_{j}(1:m,i),\boldX_{j}^{\pi}(1:m,i),\boldY_{j}^{\pi}(1:m,i)~:~ j \in [2]}$} \nonumber
\end{equation}
Let $(\boldV_{j}^{\pi}(1:m,i), \boldY_{j}^{\pi}(1:m,i)) \sim \prod p_{\hat{V}_{j}\hat{Y}_{j}}$, where $p_{\hat{V}_{j}}=p_{V_{j}}$. We wish to bound the difference $I(V_{j};Y_{j})-I(\hat{V}_{j};\hat{Y}_{j})$ from above. Using the relations $p_{\hat{U}_{j}}=p_{U}$, $p_{\hat{V}_{j}\hat{Y}_{j}|\hat{U}_{1}\hat{U}_{2}}(v,y|u,u)=p_{V_{j}Y_{j}|U}(v,y|u)$, where $\prod p_{\hat{U}_{1}\hat{U}_{2}}$ is the pmf of the interleaved vector $\boldU_{1}^{\pi}(1:m,i)\boldU_{2}^{\pi}(1:m,i)$, $\sum_{u_{1}\neq u_{2}}p_{\hat{U}_{1}\hat{U}_{2}}\leq \phi$, we can prove $|p_{V_{j}Y_{j}}(v,y)-p_{\hat{V}_{j}\hat{Y}_{j}}(v,y)|\leq \phi$. These steps are analogous to those in \cite[Appendix B]{201601arXiv_Pad}. Using \cite[Proof of Lemma 2.7]{CK-IT2011-Shrt}, we conclude $I(V_{j};Y_{j})-I(\hat{V}_{j};\hat{Y}_{j}) \leq \mathcal{L}_{C}(\phi,|\mathcal{U}|)$. We refer to reader to \cite{201701arXiv_Pad} for the rest of the proof which is quite standard.
\end{proof}

\begin{lemma}
 \label{Lem:StrictWeakerConditionsThanLC}
 The conditions stated in Thm. \ref{Thm:FirstStep} are strictly weaker than those stated in \cite[Thm. 1]{201112TIT_LiuChe-Shrt}.
\end{lemma}
\begin{proof}
Ex. \ref{Ex:DueckExampleSlightlyModified}, with $a,k$ chosen sufficiently large, satisfies conditions stated in Thm \ref{Thm:FirstStep}. In particular, choose $\delta=\frac{1}{k}, \rho=1,A=(1-\frac{1}{k^{3}})\log a, B=H(S_{1})-A,l=k^{4}a^{\frac{\eta k}{2}},\mathcal{K}=\mathcal{S}_{j}$, $f_{j}=$ identity, $\mathcal{V}_{j}=\mathcal{X}_{j}, p_{U}\mbox{ uniform}, p_{V_{j}}\mbox{ capacity achieving}$. The result now follows from Lemma \ref{Lem:ExDoesNOTSatisfyLCConditions}.
\end{proof}

\subsection{Additional information via Message-Splitting}
\label{SubSec:SuperpositionCodingOverfB-LCoding}
We now employ Han-Kobayashi technique to communicate the rest of the information (LHS of (\ref{Eqn:ConditionCentralToSeparation})). Towards that end, let $\mathscr{HK}(p_{V_{1}W_{1}X_{1}}p_{V_{2}W_{2}X_{2}})$ be the Han-Kobayashi inner bound defined in \cite[Proposition 3]{201112TIT_LiuChe-Shrt}.
\begin{thm}
 \label{Thm:SecondStep}
$(\ulineCalS,\mathbb{W}_{\ulineS})$ is transmissible over IC $\mathbb{W}_{\ulineY|\ulineX}$ if there exists (i) a finite set $\mathcal{K}$, maps $f_{j}:\mathcal{S}_{j}\rightarrow \mathcal{K}$, with $K_{j}=f_{j}(S_{j})$ for $j\in[2]$, (ii) $l \in \naturals, \delta > 0$, (iii) finite set $\mathcal{U},\mathcal{V}_{j},\mathcal{W}_{j}:j \in [2]$ and pmf $p_{U}p_{V_{1}W_{1}}p_{V_{2}W_{2}}p_{X_{1}|UV_{1}W_{1}}p_{X_{2}|UV_{2}W_{2}}$ defined on $\mathcal{U}\times \underline{\mathcal{V}}\times \underline{W}\times \ulineCalX$, where $p_{U}$ is a type of sequences in $\mathcal{U}^{l}$, (iv) $A,B \geq 0$, $\rho \in (0,A)$ such that $\phi \in [0,0.5)$, \vspace{-0.1in}
 \begin{eqnarray}
 \label{Eqn:TypicalSetSize}
A+B \geq (1+\delta)H(K_{1}),\mbox{ and for }j \in [2],\nonumber\\
\left(\substack{B+\mathcal{L}^{C}(\phi,|\underline{\mathcal{U}\mathcal{V}\mathcal{W}}|)\\ H(S_{j}|K_{1}) +\mathcal{L}^{S}_{l}(\phi,|\mathcal{S}_{j}|)}:j\in[2]\right) \in \mathscr{HK}(p_{V_{1}W_{1}X_{1}}p_{V_{2}W_{2}X_{2}}),\nonumber\\
\mathcal{L}_{j}^{C}(\phi,|\mathcal{U}|) = h_{b}(\phi)+5\phi\log |\underline{\mathcal{U}\mathcal{V}\mathcal{W}}|+||\underline{\mathcal{U}\mathcal{V}\mathcal{W}}|^{3}\phi\log\frac{1}{\phi}. \nonumber\vspace{-.9in}
\end{eqnarray}
where, $\phi,\mathcal{L}^{\scriptsize S}_{l}(\phi,|\mathcal{S}_{j}|)$ are as defined in Thm \ref{Thm:FirstStep}.
\end{thm}
\begin{remark}
 \label{Rem:HanKobayashiAchievability}For simplicity and compact description, we derive a uniform upper bound on all the mutual-information quantities involved in the description of the Han-Kobayashi region. This explains the large constant multiplying $\phi\log\frac{1}{\phi}$.
\end{remark}
Our third step is to use the decoded fixed B-L channel codewords towards conditional decoding of the outer code. The outer code is built on $\mathcal{X}_{j}$ and is superimposed on (interleaved vectors of) $C_{U}$. The challenge is that a fraction $\phi$ of the decoded codewords are erroneous. The approach is to treat the interleaved columns of the decoded $\hat{\boldU}$ as a noisy state/side information. Interleaving ensures that these sub-vectors have a S-L IID pmf. Proof is similar to \cite[Proof of Thm. 1]{201601arXiv_Pad}.

\begin{thm}
 \label{Thm:ThirdStep}
$(\ulineCalS,\mathbb{W}_{\ulineS})$ is transmissible over an IC $\mathbb{W}_{\ulineY|\ulineX}$ if there exists (i) a finite set $\mathcal{K}$, maps $f_{j}:\mathcal{S}_{j}\rightarrow \mathcal{K}$, with $K_{j}=f_{j}(S_{j})$ for $j\in[2]$, (ii) $l \in \naturals, \delta > 0$, (iii) finite set $\mathcal{U}$ and pmf $p_{U}p_{X_{1}|U}p_{X_{2}|U}$ defined on $\mathcal{U}\times \ulineCalX$, where $p_{U}$ is a type of sequences in $\mathcal{U}^{l}$, (iv) $A,B \geq 0$, $\rho \in (0,A)$ such that $\phi \in [0,0.5)$, where
 \begin{eqnarray}
 \label{Eqn:TypicalSetSize}
A+B \geq (1+\delta)H(K_{1}),\mbox{ and for }j \in [2],\nonumber\\
B+H(S_{j}|K_{1}) + \mathcal{L}^{S}(\phi,|\mathcal{K}|) < I(X_{j};\Out_{j}| U)-\mathcal{L}_{j}^{C}(\phi,|\mathcal{U}|)\nonumber\\\mathcal{L}_{j}^{C}(\phi,|\mathcal{U}|) = h_{b}(\phi)+\phi\log |\mathcal{U}|+|\mathcal{X}_{j}||\mathcal{Y}||\mathcal{U}|(1+|\mathcal{X}_{\msout{j}}|)\phi\log\frac{1}{\phi} \nonumber
 \end{eqnarray} and $\phi$, $\mathcal{L}^{\scriptsize S}(\phi,|\mathcal{K}|)$ are as defined in Thm \ref{Thm:FirstStep}.
\end{thm}
The final step in our generalization will combine the techniques of Thm. \ref{Thm:SecondStep}, \ref{Thm:ThirdStep}. In particular, we employ Han-Kobayashi technique in the superposition layer. The message to be communicated through the outer code is split into private and public parts and coded using separate codebooks. Decoder $j$ uses the decoded fixed B-L channel codeword and employs a conditional Han-Kobayashi decoder. We omit a characterization in the interest of brevity.


\bibliographystyle{../sty/IEEEtran}
{
\bibliography{../wislBib/wisl}
\end{document}

Ertem Tuncel

Paulo Minero, Denis Gunduz, Y-H Kim, Lapidoth, 

Y. Murin, R. Dabora and D. Gunduz, On joint source-channel coding for correlated sources over multiple-access relay channels, IEEE Trans. Information Theory, vol. 60, no. 10, pp. 6231–6253, Oct. 2014.

P. Minero, S.H. Lim, and Y.-H. Kim, ``A Unified Approach to Hybrid Coding'', IEEE Trans. on Information Theory, vol.61, no.4, pp.1509--1523, Apr. 2015.

A. Lapidoth and S. Tinguely, "Sending a bivariate Gaussian over a Gaussian MAC" IEEE Transactions on Information Theory, IT-56, No. 6, pp. 2714-2752, June 2010.

http://ieeexplore.ieee.org/document/7541654/

The achievable distortion region of sending a bivariate Gaussian source on the Gaussian broadcast channel
C. Tian, S. Diggavi and S. Shamai, IEEE Trans. Inform. Theory, Vol. 57, No. 10, pp. 6419-6427, Oct. 2011.

K. Khezeli and J. Chen
"A Source-Channel Separation Theorem With Application to the Source Broadcast Problem" ( pdf )
IEEE Transactions on Information Theory, vol. 62, pp. 1764-1781, Apr. 2016

L. Song, J. Chen, and C. Tian
"Broadcasting Correlated Vector Gaussians" ( pdf )
IEEE Transactions on Information Theory, vol. 61, pp. 2465-2477, May 2015

Correlated Gaussian Sources over Gaussian Weak Interference Channels

Our work takes an interesting/positive viewpoint. Can we obtain a S-L innerbound to the performance of an $l-$letter technique that is strictly larger than the current known best (which is purely based on a S-L coding technique), and that is computable? We answer this in the affirmative and develop tools to characterize an approximate performance of an $l-$letter technique via S-L expressions and enlarges upon the current known best. This, we believe, provides us a new attack strategy for multi-user problems that do not seem to admit optimal S-L coding techniques.